\documentclass[runningheads]{llncs}
\usepackage[T1]{fontenc}
\usepackage{graphicx}
\usepackage{auxlib}

\begin{document}
\title{Mechanism Design for Auctions with Externalities on Budgets\thanks{
  This research was supported by the National Natural Science Foundation of China (NSFC) under grant numbers No. 62172012, No. 12471339 and No. 62302166, 
  and the Key Laboratory of Interdisciplinary Research of Computation and Economics (SUFE), Ministry of Education.
}}

\author{
Yusen Zheng\inst{1}\and
Yukun Cheng\inst{2*}\and
Chenyang Xu\inst{3}\and
Xiaotie Deng\inst{1*}
}

\institute{
School of Computer Science, Peking University, Beijing 100871, China
\and
School of Business, Jiangnan University, Wuxi 214122, China
\and
Software Engineering Institute, East China Normal University, China\\
\email{yusen@stu.pku.edu.cn, ykcheng@amss.ac.cn, cyxu@sei.ecnu.edu.cn, xiaotie@pku.edu.cn}
\\
$^{*}$Corresponding author
}

\maketitle              %

\begin{abstract}

This paper studies mechanism design for auctions with externalities on budgets, a novel setting where the budgets that bidders commit are adjusted due to the externality of the competitors' allocation outcomes---a departure from traditional auctions with fixed budgets.
This setting is motivated by real-world scenarios, for example, participants may increase their budgets in response to competitors' obtained items. 
We initially propose a general framework with homogeneous externalities to capture the interdependence between budget updates and allocation, formalized through a budget response function that links each bidder's effective budget to the amount of items won by others.

The main contribution of this paper is to propose a truthful and individual rational auction mechanism for this novel auction setting, which achieves an approximation ratio of $1/3$ with respect to the liquid welfare.
This mechanism is inspired by the uniform-price auction, in which an appropriate uniform price is selected to allocate items, ensuring the monotonicity of the allocation rule while accounting for budget adjustments.
Additionally, this mechanism guarantees a constant approximation ratio by setting a purchase limit.
Complementing this result, we establish an upper bound: no truthful mechanism can achieve an approximation ratio better than $1/2$.
This work offers a new perspective to study the impact of externalities on auctions, providing an approach to handle budget externalities in multi-agent systems.

\end{abstract}

\newpage

\section{Introduction}\label{sec:introduction}

Auctions are a common method for resource allocation. Bidders participating in auctions typically face constraints on the amount of money they can use.
Many previous studies on auctions with budget constraints model the budget as a predetermined fixed parameter~\cite{abrams2006revenue,daskalakis2018revenue,richter2019mechanism,devanur2013prior,dobzinski2014efficiency}. 
Recently, some work has begun to explore scenarios where bidders do not set a hard budget parameter but instead focus on their return on investment (ROI)~\cite{golrezaei2021auction,goel2014clinching,balseiro2021robust}.
However, in real-world economic activities, bidders' budget decisions are often influenced by social signals, such as the behavior of others.
This phenomenon is exemplified by concepts like the ``herding effect'' or ``conspicuous consumption'' commonly studied in behavioral economics. 
In these cases, these budget constraints are not rigid but rather flexible, meaning that bidders' internal expectations of their budget ceilings continuously adjust throughout the auction process and its outcomes. Auction mechanisms designed for fixed budget scenarios are not suitable for situations with such externalities on the budget. This paper formalizes the problem considering the specific flexible budget constraint for the first time, in which a bidder's budget is influenced by the final resource allocations obtained by others. This scenario is quite common. For example, bidders exhibiting strong herding behavior may psychologically increase their budget ceilings upon observing that other bidders receive substantial resource allocations.

We formally refer to the impact of allocation outcomes on bidders' budgets as \emph{the allocation(-induced) externalities on budgets}. 
Previous works on externalities primarily focus on the effects of item allocations or the private information held by bidders on other bidders' valuations, such as~\cite{agarwal2024towards,belloni2017mechanism,krysta2010combinatorial,jehiel2005allocative,jehiel1996not}. 
However, our study examines the impact of item allocation on bidders' budgets, representing a new category of externality problems.

The presence of allocation externalities on budgets means that the allocation of items is no longer an isolated decision for each bidder, rather, it simultaneously affects the utilities and budgets of all bidders, including the bidder herself. This introduces new challenges in designing effective auction mechanisms. 
Our goal, therefore, is to develop an auction mechanism that not only provides robust incentive guarantees but also maximizes efficiency.

\subsection{Our Contributions}

We formalize the allocation externalities on budgets by introducing a budget impact factor, which captures these externalities. We consider a scenario with homogeneous externalities, meaning that for every bidder, the externality effect of other bidders' allocations on their budget is identical. We use liquid welfare as the metric for measuring auction efficiency, as it is a commonly used welfare measure in budget-constrained auctions~\cite{dobzinski2014efficiency}. Optimizing liquid welfare involves allocating items to bidders with high valuations while ensuring they possess sufficient purchasing power.

Compared to traditional welfare-maximizing auction mechanisms, the presence of allocation externalities on budgets introduces new challenges. 
Optimizing liquid welfare requires balancing valuation and budget, which can be in conflict. 
For instance, allocating more items to one bidder increases her value, but simultaneously reduces the number of items available to others, which in turn diminishes that bidder's budget due to the externalities.
This conflict, combined with incentive constraints, necessitates novel technical approaches in mechanism design.
Overall, our contributions are summarized as follows:

\begin{itemize}[left=1em]
  \item We propose a model to capture allocation-induced externalities on budgets in auctions. 
  In this model, a bidder's budget is represented as a function of the allocations received by other bidders, allowing us to characterize the impact of homogeneous allocation externalities on the budget.\vspace{0.5em}
  \item We provide an optimal allocation algorithm that solves the liquid welfare maximization problem (\cref{thm:opt-allocation}).
  The algorithm greedily allocates items in decreasing order of valuation and assigns any remaining items to the bidder whose budget is least affected by externalities (\cref{alg:allocation}).
  \vspace{0.5em}
  \item We design a truthful and individually rational auction mechanism (\cref{mech:approx-lw}) that achieves a $1/3$-approximation for liquid welfare (\cref{thm:approx-lw}). Furthermore, we prove an impossibility result showing that no truthful mechanism can achieve a better approximation ratio than $1/2$ (\cref{sol:approx-lw:upper-bound}).
\end{itemize}

\subsection{Related Work}

\paragraph{Mechanism Design with Budget Constraints.}

Budget constraints are a natural feature of economic activities, limiting the amount of money participants can spend. Many studies have focused on designing auction mechanisms that account for these constraints~\cite{abrams2006revenue,daskalakis2018revenue,richter2019mechanism,devanur2013prior}.
The most common approach models budget constraints as hard constraints, treating a bidder's budget as a fixed parameter, and uses budgeted quasi-linear utility functions to characterize bidders' utilities. This approach has been widely applied in revenue-maximizing auctions~\cite{abrams2006revenue,daskalakis2018revenue}, welfare-maximizing auctions~\cite{dobzinski2012multi,dobzinski2014efficiency,lu2015improved}, multi-item auctions~\cite{lu2017liquid}, and auctions in auto-bidding systems~\cite{balseiro2022optimal}.
This model is reasonable in many scenarios, such as certain resource allocation settings where bidders' budgets remain fixed.
However, in many real-world scenarios, bidders' budgets are not rigid but rather flexible, adjusting throughout the auction process and its outcomes.
Only a few studies have explored auction mechanism design under such conditions. For instance, Goel et al.~\cite{goel2014clinching} model the budget as a function of the items acquired by the bidder, but do not consider the impact of other bidders' allocations on the budget. 
The study of auctions with externalities on budget remains a novel area of research.

\paragraph{Efficiency Metrics in Auctions with Budget Constraints.}

In their pioneering and comprehensive study of auctions with budget constraints, Dobzinski et al.~\cite{dobzinski2012multi} use \emph{Pareto efficiency} as a measure of mechanism performance. They introduced the adaptive clinching auction mechanism and proved that it is the only mechanism that simultaneously satisfies truthfulness, individual rationality, and Pareto optimality. Subsequent works on auctions with budget constraints have also frequently relied on Pareto efficiency as a measure of efficiency~\cite{fiat2011single,hafalir2012near,goel2014clinching}. However, since Pareto efficiency is a binary criterion, it does not lend itself to the development of approximation algorithms.
Another commonly used metric for measuring auction efficiency, \emph{social welfare}, has been shown to achieve no better than a $1/n$-approximation under budget constraints for any truthful and individually rational mechanism~\cite{dobzinski2014efficiency}. Recognizing these limitations, Dobzinski et al.~\cite{dobzinski2014efficiency} introduced a more quantifiable welfare metric, \emph{liquid welfare}, which innovatively incorporates both agents' valuations of items and their purchasing power.
They further leveraged this metric to design a mechanism that achieves a $1/2$-approximation ratio. Building on this work, Lu and Xiao~\cite{lu2015improved} improved the approximation ratio to $\frac{2}{1 + \sqrt{5}}$. Many subsequent studies have used liquid welfare as a metric to evaluate the efficiency of mechanisms under budget constraints~\cite{fikioris2023liquid,fotakis2019bridge,azar2017liquid}.

\paragraph{Auctions with Externalities.}

The study of externalities in economics has a long history. In 1920, Pigou proposed using taxes to correct market failures caused by externalities~\cite{pigou2017economics}. In 1960, Coase introduced the famous ``Coase Theorem'', emphasizing that, in the absence of transaction costs, parties can resolve externality issues through negotiation without relying on government intervention~\cite{coase2013problem}. 
Externality problems have also been extensively studied in the field of algorithmic mechanism design, particularly in the context of auction mechanisms.
In the interesting work by Jehiel et al.~\cite{jehiel1996not}, they discuss how to design auction mechanisms to maximize the auctioneer's revenue when the allocation of items negatively impacts the utilities of other bidders. This is one of the earliest studies analyzing \emph{allocation externalities}.
The survey by Jehiel et al.~\cite{jehiel2005allocative} provides a comprehensive analysis of various novel mechanism design problems arising from the presence of externalities, including allocation externalities and \emph{information externalities}. 
Krysta et al.~\cite{krysta2010combinatorial} investigated externality issues in multi-item combinatorial auctions using a bidding language framework.
Agarwal et al.~\cite{agarwal2024towards} studied the impact of additive externalities in data auctions.
In addition to analyzing the impact of externalities on auctions themselves, there is also work that explores how to design the structures of these externality effects. Belloni et al.~\cite{belloni2017mechanism} investigated the novel problem of how monopolistic sellers can maximize their revenue by designing network structures that capture the externality effects among buyers in the market.
However, to date, no work has explored the impact of allocation-induced externalities on budgets.

\subsection{Paper Organization}
In \cref{sec:preliminaries}, we introduce a model characterizing the impact of externalities on budgets.
In \cref{sec:optimal}, we present an allocation algorithm maximizing the liquid welfare and prove its optimality.
In \cref{sec:mechanism}, we design a truthful and individually rational auction mechanism that achieves a $1/3$-approximation for the liquid welfare.
Finally, we conclude this paper in \cref{sec:conclusion}.

\section{Preliminaries}\label{sec:preliminaries}

Let us consider a sealed-bid auction for one divisible item. 
There are $n$ bidders indexed by $i \in \{1, 2, \ldots, n\}$, and we assume that $n \ge 2$, as with fewer bidders (i.e., $n=1$), there are no externalities. Each bidder $i$
has two parameters: a private \emph{valuation} $v_i$ for one unit of the item, known only to bidder $i$;
and a public \emph{budget impact factor} $\alpha_i\in (0,+\infty)$\footnote{As we will explain shortly, if $\alpha_i = 0$, this implies that bidder $i$'s budget is zero. Consequently, we can safely exclude this bidder from consideration.}, which quantifies how bidder $i$'s budget is influenced by the allocations to other bidders. (This will be formalized later.)
An \emph{auction mechanism} $\cM$ is an algorithm that takes the bidders' reported valuations, denoted by ${\bf b}=(b_1,\cdots,b_n)$, and the public budget impact factors $(\alpha_1,\cdots,\alpha_n)$ as input and outputs a tuple $({\bf x},{\bf p})$. To be specific, ${\bf x}=(x_1,\cdots,x_n)$ is the allocation of item, where $x_i\in [0,1]$, denotes the amount of the item allocated to bidder $i$, satisfying $\sum_{i}x_i\leq 1$.
Payment ${\bf p}=(p_1,\cdots,p_n)$, where $p_i$ is the amount bidder $i$ shall pay. 
We slightly abuse notation by using the functions $x_i(\vecb)$ and $p_i(\vecb)$ to denote the fraction of item allocated to bidder $i$ and the corresponding payment, respectively, by mechanism $\cM$ when the input bid profile is $\vecb$.
For the sake of convenience, the bid profile ${\bf b}$ and the valuation profile ${\bf v}$ can be equivalently denoted as ${\bf b}=(b_i,{\bf b}_{-i})$ and ${\bf v}=(v_i,{\bf v}_{-i})$.

Unlike the traditional auction setting with hard budget constraints, we are more interested in a scenario, where bidders' budgets are subject to allocation-induced externalities. 
Specifically, a bidder's budget is influenced linearly by the allocations received by other bidders. 
Thus, given the item allocation $\vecx = (x_1,\ldots, x_n)=(x_i,\vecx_{-i})$, the \emph{budget} of bidder $i$ is determined as $B_i(\vecx_{-i}) = \alpha_i \cdot \sum_{j \in [n] - i}  x_j$, 
in which, each bidder $i$ experiences a homogeneous impact from the allocations of other bidders, as captured by the parameter $\alpha_i$.   

We employ a budgeted quasi-linear utility function to describe a bidder's utility. Given a reported bidding profile $\vecb$, and thus bidder $i$ obtains her item allocation $x_i(\vecb)$ with payment $p_i(\vecb)$ from a mechanism. Then the utility of bidder $i$ is defined as
\begin{equation*}
  u_i(\vecb;v_i) = \begin{cases}
    v_ix_i(\vecb)-p_i(\vecb), &p_i(\vecb)\le B_i(\vecx_{-i}(\vecb)) \\
    -\infty, &\text{otherwise}
  \end{cases}.
\end{equation*}

A mechanism is truthful if, for each bidder reporting their true valuations is always a dominant strategy. This essential property in mechanism design ensures that a utility maximizer will always report their true valuation to the mechanism.
\begin{definition}[Truthfulness]
    An auction mechanism $\cM$ is said to be \emph{truthful}, if $
    u_i(v_i,\vecb_{-i};v_i)\geq u_i(b_i,\vecb_{-i};v_i)$ for any $i\in [n]$, $v_i$, $b_i$ and $\vecb_{-i}$.
\end{definition}

A mechanism is individually rational (IR) if, for each bidder, bidding truthfully results in non-negative utility, regardless of how other bidders bid.

\begin{definition}[Individual Rationality]
  An auction mechanism is \emph{individually rational} if $u_i(v_i,\vecb_{-i};v_i)\geq 0$, for any $i\in [n]$, $v_i$ and $\vecb_{-i}$.
\end{definition}

Inspired by \cite{dobzinski2014efficiency}, we employ \emph{liquid welfare} as the welfare measure, to evaluate the efficiency of our designed auction mechanism.
\begin{definition}[Liquid Welfare]
    In the auction with externalities on budgets, the liquid welfare, associated an allocation $\vecx$, is defined as
    \begin{equation*}
  \lw(\vecx)=\sum_{i\in[n]}\min\{v_ix_i,B_i(\vecx_{-i})\}.
\end{equation*}
\end{definition}
Optimizing liquid welfare reflects the idea of allocating the item to bidders who value it highly and have sufficient purchasing power.

For any instance $I=\{v_i,\alpha_i\}_{i \in [n]}$, let $\opt(I)$ denote the optimal liquid welfare achievable by any optimal allocation.
We say that a truthful auction mechanism is a \emph{$\rho$-approximation} with respect to the liquid welfare if, for any instance $I=\{v_i,\alpha_i\}_{i \in [n]}$, the liquid welfare achieved by the mechanism is at least $\rho \cdot \opt(I)$, where $\rho \in (0,1]$ is the \emph{approximation ratio}.

Our goal is to design a truthful and individually rational auction mechanism that achieves a good approximation ratio with respect to the liquid welfare.

\section{Optimal Allocation for Maximizing Liquid Welfare}\label{sec:optimal}

This section focuses on the design of an allocation algorithm that takes into account the externalities on budgets in order to maximize liquid welfare, while temporarily disregarding the incentive constraints of truthfulness and individual rationality.
We begin with presenting a greedy algorithm (\cref{alg:allocation}) and subsequently prove that it achieves the optimal liquid welfare (\cref{thm:opt-allocation}).

\begin{algorithm}[!ht]
  \caption{Optimal Allocation}
  \label{alg:allocation}
  \KwIn{Number of bidders $n$, valuations $\set{v_i}_{i\in[n]}$ and  budget impact factors $\set{\alpha_i}_{i\in [n]}$ \tcp*{Assume $v_1\ge v_2\ge\cdots\ge v_n$.}} 
  \KwOut{Allocation $\vecx^*=(x_1^*,x_2^*,\cdots,x_n^*)$\;}
  
  $s\gets 1$ \tcp*{$s$ represents the remaining quantity of the item.}
  \For{$i=1..n$}{
    \If{$s\ge \frac{\alpha_i}{v_i+\alpha_i}$}{
      $x_i^*\gets \frac{\alpha_i}{v_i+\alpha_i}$, $s\gets s-x_i^*$ \;
    }\ElseIf{$s \ge 0$}{
      $x_i^*\gets s$, $s\gets 0$ \;
    }
  }
  \If{$s>0$}{
    $\l \gets \argmin_i \alpha_i$, $x_\l^*\gets x_\l^*+s$, $s\gets 0$ \;
  }
  \Return $(x_1^*,x_2^*,\cdots,x_n^*)$ \;
\end{algorithm}

In brief, \cref{alg:allocation} first allocates to each bidder $i$ a fraction of the item equal to $\min\set*{\frac{\alpha_i}{v_i+\alpha_i}, s}$ in descending order of valuations, where $s$ denotes the remaining fraction of the item. 
After these operations, if there is a 
fraction of item left, then the residual fraction is assigned to bidder $\ell$, whose budget is least affected by externalities, that is $\alpha_{\ell}=\min_{i\in [n]}\alpha_i$.  
Note that \cref{alg:allocation} distributes the entire item. Specially, if one bidder $i$ receives $\frac{\alpha_i}{v_i+\alpha_i}$ of the item, it implies that her budget, given by $B_i(\vecx_{-i})=\alpha_i\cdot\sum_{j\in[n]-i}x_j=\alpha_i\cdot(1-x_i)$, is exactly equal to the value $v_i x_i$.
There are two scenarios for the outcome of \cref{alg:allocation}:
\begin{itemize}
  \item \emph{Case 1:} $\sum_{i\in [n]} \frac{\alpha_i}{v_i+\alpha_i}\ge 1$. Let $r\deq \max\set*{r\in [n] \mid \sum_{i=1}^r \frac{\alpha_i}{v_i+\alpha_i} \le 1}$. Then $x_i^*=\frac{\alpha_i}{v_i+\alpha_i}$ for any $i<r$, $x_{r+1}^*=1-\sum_{i=1}^rx_i^*$ and $x_i^*=0$ for any $i>r+1$.
  \vspace{0.5em}
  \item \emph{Case 2:} $\sum_{i\in [n]} \frac{\alpha_i}{v_i+\alpha_i}<1$. Let $\l\deq \argmin_i \alpha_i$. Then $x_i^*=\frac{\alpha_i}{v_i+\alpha_i}$ for any $i\neq \l$, and $x_{\l}^*=1-\sum_{i\neq \l}x_i^*$.
\end{itemize}

\begin{theorem}\label{thm:opt-allocation}
  The allocation returned by \cref{alg:allocation} is an optimal solution for maximizing the liquid welfare.
\end{theorem}

Due to space constraints, we provide a proof sketch in the following. The detailed proof of \cref{thm:opt-allocation} is available in \cref{sec:proof-optimal-thm}.

\begin{proof}[Proof Sketch]
  We prove this theorem by initially demonstrating that there exists an optimal allocation $\vecx = (x_1, x_2, \cdots, x_n)$ that satisfies the following properties:

  \begin{itemize}[left=2em]
    \item[(P1)] All items are allocated, 
    i.e., $\sum_{i\in[n]}x_i=1$.
    \vspace{0.5em}
    \item[(P2)] No overallocation occurs other than for bidder $\l$: 
    For every bidder $i$ with $i \neq \l$, the allocation satisfies $x_i \le \frac{\alpha_i}{v_i + \alpha_i}$.
    \vspace{0.5em}
    \item[(P3)] Allocation follows the order:
    For any pair of bidders $i$ and $j$ with $1 \leq i < j \leq n$, either $x_i \geq \frac{\alpha_i}{v_i + \alpha_i}$ or $x_j = 0$.
    \vspace{0.5em}
    \item[(P4)] Overallocation for bidder $\ell$ only occurs when the allocations for all other bidders are satisfied:
    For every bidder $i$ with $i\neq \l$, the allocation satisfies either $x_\l \le \frac{\alpha_\l}{v_\l+\alpha_\l}$ or $x_i \ge \frac{\alpha_i}{v_i+\alpha_i}$.
  \end{itemize}

  We then prove that there is only one allocation satisfying (P1)-(P4), which matches exactly with the allocation provided by \cref{alg:allocation}. Therefore, \cref{alg:allocation} is optimal.
\end{proof}

\section{A $1/3$-Approximation Truthful Mechanism}\label{sec:mechanism}

\cref{alg:allocation} provides an optimal allocation for maximizing liquid welfare just from the perspective of algorithm design. However, this allocation rule lacks monotonicity, since the allocation for the bidder with the highest valuation decreases as her valuation increases. This violates truthfulness as per Myerson's Lemma~\cite{myerson1981optimal}, which we have detailed in the appendix (\cref{lem:myerson}).

Fortunately, we can maintain the monotonicity of the allocation rule by modifying the allocation rule, while only incurring a constant loss in the approximation ratio for liquid welfare.
Our mechanism design is inspired by the uniform-price auction in the work of \cite{dobzinski2014efficiency}.
However, the external impact of allocation on budgets makes our problem setting more complex. To achieve an ideal approximation ratio while ensuring the desirable properties of truthfulness and individual rationality, we need to explore new mechanism design techniques. Our main result is presented in \cref{thm:approx-lw}.

\begin{theorem}\label{thm:approx-lw}
 For the auction with externalities on budgets, there exists a mechanism that ensures truthfulness, individual rationality, and achieves an approximation ratio of $1/3$ with respect to liquid welfare.
\end{theorem}

\subsection{Mechanism Design}

We present the auction mechanism (\cref{mech:approx-lw}) and design technique in detail.

\begin{mechanism}[Uniform Price Auction with Purchase Limit]\label{mech:approx-lw}
  Consider $n$ bidders with valuations $v_1 \ge v_2 \ge \cdots \ge v_n$ and budget impact factors $\{\alpha_1, \cdots, \alpha_n\}$. \cref{mech:approx-lw} proceeds as follows: 
  First, add a dummy bidder $n+1$ with valuation $v_{n+1} = 0$ and budget impact factor $\alpha_{n+1} > 0$. 
  Next, define 
  $$k \deq \max\set*{\l \in [n] \mid \sum_{i=1}^\l \min\left\{\frac{\alpha_i}{v_\l + \alpha_i}, \frac{1}{2}\right\} \le 1},$$ referred to as the \emph{division point}, and let $q$ denote the smallest\footnote{There may be multiple roots, and choosing any of them is acceptable, as the choice does not affect the value of $\min\left\{\frac{\alpha_i}{q + \alpha_i}, \frac{1}{2}\right\}$. For simplicity, we select the smallest one.} non-negative root of the equation $\sum_{i=1}^k \min\left\{\frac{\alpha_i}{q + \alpha_i}, \frac{1}{2}\right\} = 1$, referred to as the \emph{unifrom price}.
  The allocation is determined based on the following two cases:

  \begin{itemize}
    \item If $q > v_{k+1}$, allocate $x_i=\min\set*{\frac{\alpha_i}{q + \alpha_i}, \frac{1}{2}}$ to each bidder $i \le k$, and allocate 0 to the remaining bidders.
    \vspace{0.5em}
    \item If $q \le v_{k+1}$, allocate $x_i=\min\set*{\frac{\alpha_i}{v_{k+1} + \alpha_i}, \frac{1}{2}}$ to each bidder $i \le k$, allocate $x_{k+1}=1 - \sum_{j \in [k]} x_j$ to bidder $k+1$, and allocate 0 to the remaining bidders.
  \end{itemize}
 Payment $p_i$ for each bidder $i$ is calculated according to Myerson's Lemma (\cref{lem:myerson}):
  \begin{equation*}
    p_i = v_i x_i(v_i,\vecv_{-i}) - \int_{0}^{v_i} x_i(z,\vecv_{-i}) \dif z,
  \end{equation*}
  where $x_i(\vecv)$ denotes the fraction of item allocated to bidder $i$ by Mechanism 1, if the reported information is $\vecv$.
  
\end{mechanism}

The dummy bidder $n+1$ added in \cref{mech:approx-lw} can help us handle the case where $k=n$.
As shown in \cref{sol:approx-lw:dummy}, adding this bidder does not affect the allocation for the ``actual'' bidders $\set*{1,2,\ldots,n}$, nor does it impact the liquid welfare objective. This is because the item assigned to the dummy bidder is $x_{n+1}=0$. Consequently, her contribution to liquid welfare is zero, and she also cannot influence the contributions of other bidders through externalities.
Due to space limitations, we defer the proof to \cref{proof:approx-lw:dummy}.

\begin{lemma}\label{sol:approx-lw:dummy}
  In \cref{mech:approx-lw}, the allocation to dummy bidder is zero, i.e., $x_{n+1}=0$, and the division point $k$ satisfies $k \leq n$.
\end{lemma}

Due to \cref{sol:approx-lw:dummy}, it is clear that \cref{mech:approx-lw} allocates the entire item to bidder $\set*{1,2,\cdots,n}$. Therefore, $B_i(\vecx_{-i})= \alpha_i \cdot (1 - x_i)$ for all $i \in [n]$.
For the case of $q \le v_{k+1}$, bidder $k+1$ is allocated 
$x_{k+1}=1-\sum_{i\in [k]} \min\set*{\frac{\alpha_i}{v_{k+1}+\alpha_i},\frac{1}{2}}\ge 1-\sum_{i\in[k]} \min\set*{\frac{\alpha_i}{q+\alpha_i},\frac{1}{2}}=0$.
It implies that the allocation for bidder $k+1$ is non-negative. 
Moreover, by the definition of $k$, we have 
$\sum_{i\in[k+1]} \min\set*{\frac{\alpha_i}{v_{k+1}+\alpha_i},\frac{1}{2}} > 1$.
Therefore, 
\begin{equation}\label{eq:approx-lw:xkplus1}
  0\le x_{k+1} < \min\set*{\frac{\alpha_{k+1}}{v_{k+1}+\alpha_{k+1}},\frac{1}{2}}
\end{equation}

\paragraph{The Role of the Purchase Limit.}
\cref{mech:approx-lw} ensures that each bidder receives no more than \(1/2\) of the item, which can be viewed as a form of \emph{purchase limit}.
The purchase limit is essential for the mechanism to achieve a good approximation ratio. The high-level intuition is as follows: Consider a bidder with a budget impact factor \( \alpha \) and a valuation \( v = \alpha^2 \), where \( \alpha \) is sufficiently large for this bidder to rank first in the valuation order. The uniform price is \( q \). 
Without the purchase limit, this bidder would receive \( x = \frac{\alpha}{q + \alpha} \), and her contribution to liquid welfare would be \( \Gamma \deq \alpha \cdot (1 - x) = \frac{q \alpha}{q + \alpha} \). However, in the optimal allocation (\cref{alg:allocation}), the bidder receives \( x^* = \frac{\alpha}{v + \alpha} = \frac{1}{\alpha + 1} \). Her contribution to optimal liquid welfare is \( \Gamma^* \deq \alpha \cdot (1 - x^*) = \frac{\alpha^2}{\alpha + 1} \). 
As \( \alpha \to \infty \), \( \Gamma \to q \), but \( \Gamma^* \to \infty \), leading to an unbounded welfare gap. We will prove later that by imposing a purchase limit of $1/2$, the welfare gap can be bounded by a constant factor (see the proof of \cref{sol:approx-lw:approx}).

\subsection{Theoretical Guarantees of the Mechanism}

To prove that \cref{mech:approx-lw} satisfies the desired attributes in \cref{thm:approx-lw}, we first demonstrate that the allocation rule satisfies the \emph{monotonicity}, which ensures the 
truthfulness of \cref{mech:approx-lw} by Myerson's Lemma, and establish the property of \emph{budget feasibility}  to ensure that the payments returned by \cref{mech:approx-lw} are all subject to the budgets.
Due to space limits, the proofs of \cref{sol:approx-lw:monotonic} and \cref{sol:approx-lw:budget} are left in \cref{proof:approx-lw:monotonic} and \cref{proof:approx-lw:budget}, respectively.

\begin{lemma}[Monotonicity]\label{sol:approx-lw:monotonic}
  The allocation function of \cref{mech:approx-lw} is non-decreasing in the valuation of each bidder.
  Specifically, for any bidder $j$, $x_j(v_j, \vecv_{-j})$ is non-decreasing in $v_j$ for any $\vecv_{-j}$.
\end{lemma}

\begin{lemma}[Budget Feasibility]\label{sol:approx-lw:budget}
  \cref{mech:approx-lw} is budget-feasible, that is $p_i\leq B_i$, for any $i\in [n]$.
\end{lemma}

Despite the properties of monotonicity and budget feasibility, we also evaluate the efficiency of \cref{mech:approx-lw} by proving its constant approximation ratio.

\begin{lemma}[$1/3$-Approximation]\label{sol:approx-lw:approx}
  \cref{mech:approx-lw} achieves an approximation ratio of $1/3$ with respect to the liquid welfare.
\end{lemma}

\begin{proof}

For an instance $I=\{v_i,\alpha_i\}_{i\in [n]}$, let $\alg$ denote the liquid welfare achieved by \cref{mech:approx-lw}, and let $\opt$ denote its optimal liquid welfare (for instance, provided by \cref{alg:allocation}).
The corresponding allocations are denoted by $(x_1, x_2, \cdots, x_n)$ and $(x_1^*, x_2^*, \cdots, x_n^*)$, respectively.
We discuss following two cases.

\noindent{\bf Case 1:} $q>v_{k+1}$.
Define two sets $A\deq \set*{i\in[k]\mid \frac{\alpha_i}{q+\alpha_i}\le\frac{1}{2}}$ and $B\deq [k]-A$.
Thus, $q\ge \alpha_i$ for $i\in A$ and $q<\alpha_i$ for $i\in B$.
Moreover, $x_i=\frac{\alpha_i}{q+\alpha_i}$ for $i\in A$ and $x_i=1/2$ for $i\in B$.
Therefore we have
\begin{equation*}
  1=\sum_{i\in [n]} x_i = \sum_{i\in [k]} \min\set*{\frac{\alpha_i}{q+\alpha_i},\frac{1}{2}}
  =\sum_{i\in A} \frac{\alpha_i}{q+\alpha_i} + \sum_{i\in B} \frac{1}{2}.
\end{equation*}
Simplifying, we get 
\begin{equation}\label{eq:approx-lw:case1:sumA}
  \sum_{i\in A} \frac{\alpha_i}{q+\alpha_i} = 1 - \frac{\card{B}}{2}.
\end{equation}

Since $q\le v_k\le v_i$, we have $v_ix_i\ge \alpha_i(1-x_i)$ for $i\in A$.
Thus,
\begin{equation}\label{eq:approx-lw:case1:alg}
  \begin{aligned}
    \alg
    =& \sum_{i\in[n]} \min\set*{v_ix_i,\alpha_i(1-x_i)} 
    = \sum_{i\in A} \alpha_i(1-x_i) + \sum_{i\in B} \min\set*{\frac{v_i}{2},\frac{\alpha_i}{2}} 
    \\
    =& q \cdot \sum_{i\in A} \frac{\alpha_i}{q+\alpha_i} + \frac{1}{2} \cdot \sum_{i\in B} \min\set{v_i,\alpha_i} \\
    =& q \cdot \left(1-\frac{\card{B}}{2}\right) + \frac{1}{2} \cdot \sum_{i\in B} \min\set{v_i,\alpha_i} \\
    =& q + \frac{1}{2} \cdot \sum_{i\in B} \left(\min\set{v_i,\alpha_i}-q\right) \ge q > v_{k+1}.
  \end{aligned}
\end{equation}
The penultimate equality holds because of \eqref{eq:approx-lw:case1:sumA}.
The penultimate inequality holds because $q < \alpha_i$ for $i \in B$.
The last inequality holds because $q > v_{k+1}$, as the condition in Case 1.

For all $i\in A$, since $q\ge \alpha_i$, we have
\begin{equation*}
\begin{aligned}
  \frac{\left(\frac{v_i\alpha_i}{v_i+\alpha_i}\right)}{ \left(\frac{q\alpha_i}{q+\alpha_i}\right)}
  = \frac{\left(\frac{v_i}{v_i+\alpha_i}\right)}{\left(\frac{q}{q+\alpha_i}\right)} \le \frac{1}{\left(\frac{q}{q+q}\right)} = 2.
\end{aligned}
\end{equation*}
Therefore, 
\begin{equation}\label{eq:approx-lw:case1:vi}
  \sum_{i\in A} \frac{v_i\alpha_i}{v_i+\alpha_i}
  \le \sum_{i\in A} \left(2 \cdot \frac{q\alpha_i}{q+\alpha_i}\right)
  = 2 \cdot q \cdot \left(1-\frac{\card{B}}{2}\right).
\end{equation}
The last equality holds because of \eqref{eq:approx-lw:case1:sumA}.

Now we consider $\opt$:
\begin{equation*}
  \begin{aligned}
    \opt = & \sum_{i\in[n]} \min\set*{v_ix_i^*,\alpha_i(1-x_i^*)} 
    \le \sum_{i\in[k]} \frac{v_i\alpha_i}{v_i+\alpha_i} + \sum_{i=k+1}^n (v_i x_i^*)
    \\
    \le& \sum_{i\in[k]} \frac{v_i\alpha_i}{v_i+\alpha_i} + v_{k+1} = \sum_{i\in A} \frac{v_i\alpha_i}{v_i+\alpha_i} + \sum_{i\in B} \frac{v_i\alpha_i}{v_i+\alpha_i} + v_{k+1} \\
    \le& 2 \cdot \left(q\cdot\left(1-\frac{\card{B}}{2}\right) + \frac{1}{2} \cdot \sum_{i\in B}\min\set{v_i,\alpha_i}\right) + v_{k+1} \\
    \le& 2\cdot \alg + \alg = 3\cdot\alg.
  \end{aligned}
\end{equation*}
The third inequality holds because of \eqref{eq:approx-lw:case1:vi} and the fact that $\frac{v_i\alpha_i}{v_i+\alpha_i}\le \min\{v_i,\alpha_i\}$ for all $i\in [n]$.
The last equality holds because of the penultimate equality and the final inequality in \eqref{eq:approx-lw:case1:alg}.

\noindent{\bf Case 2:}  $q\le v_{k+1}$.
For simplicity, we reuse the symbols and redefine the sets $A$ and $B$ as follows:
$A\deq \set*{i\in[k]\mid \frac{\alpha_i}{v_{k+1}+\alpha_i}\le\frac{1}{2}}$, and $B\deq [k]-A$.
Thus, $v_{k+1}\ge \alpha_i$ for $i\in A$ and $v_{k+1}<\alpha_i$ for $i\in B$.
Moreover, $x_i=\frac{\alpha_i}{v_{k+1}+\alpha_i}$ for $i\in A$ and $x_i=1/2$ for $i\in B$.

According to \eqref{eq:approx-lw:xkplus1}, $x_{k+1}\le \frac{1}{2}$ and $x_{k+1}\le \frac{\alpha_{k+1}}{v_{k+1}+\alpha_{k+1}}$.
Thus $v_{k+1}x_{k+1}\le \alpha_{k+1}(1-x_{k+1})$.

Since $v_{k+1}\le v_i$, we have $v_ix_i\ge \alpha_i(1-x_i)$ for $i\in A$. So,
\begin{equation}\label{eq:approx-lw:case2:alg}
  \begin{aligned}
    \alg =& \sum_{i\in A} \alpha_i(1-x_i) + \sum_{i\in B} \min\set*{\frac{v_i}{2},\frac{\alpha_i}{2}} +v_{k+1}x_{k+1} \\
    =& \sum_{i\in A} \alpha_i \cdot \left(1-\frac{\alpha_i}{v_{k+1}+\alpha_i}\right) + \frac{1}{2} \cdot \sum_{i\in B} \min\set{v_i,\alpha_i} \\
    &\quad + v_{k+1} \cdot \left(1-\sum_{i\in A}\frac{\alpha_i}{v_{k+1}+\alpha_i}-\sum_{i\in B}\frac{1}{2}\right)\\
    =&v_{k+1}+\frac{1}{2}\cdot \sum_{i\in B}(\min\set{v_i,\alpha_i}-v_{k+1}) \ge v_{k+1}.
  \end{aligned}
\end{equation}
The last inequality holds because $v_{k+1}<\alpha_i$ for $i\in B$.
Furthermore, from the second equality in \eqref{eq:approx-lw:case2:alg} and the non-negativity of $x_{k+1}$ as shown in \eqref{eq:approx-lw:xkplus1}, we have
\begin{equation}\label{eq:approx-lw:case2:alg2}
  \begin{aligned}
    \alg \ge \sum_{i\in A} \alpha_i \cdot \left(1-\frac{\alpha_i}{v_{k+1}+\alpha_i}\right) + \frac{1}{2} \cdot \sum_{i\in B} \min\set{v_i,\alpha_i}.
  \end{aligned}
\end{equation}

Similarly to Case 1, for all $i \in A$, since $v_{k+1} \ge \alpha_i$, we have
\begin{equation}\label{eq:approx-lw:case2:vi}
  \begin{aligned}
    \frac{\left(\frac{v_i\alpha_i}{v_i+\alpha_i}\right)}{\alpha_i \cdot \left(1 - \frac{\alpha_i}{v_{k+1}+\alpha_i}\right)} 
    = \frac{\left(\frac{v_i}{v_i+\alpha_i}\right)}{\left(\frac{v_{k+1}}{v_{k+1}+\alpha_i}\right)} \le \frac{1}{\left(\frac{v_{k+1}}{v_{k+1}+v_{k+1}}\right)} = 2.
  \end{aligned}  
\end{equation}

Now we consider $\opt$:
\begin{equation*}
  \begin{aligned}
    \opt \le& \sum_{i\in[k]} \frac{v_i\alpha_i}{v_i+\alpha_i} + v_{k+1} = \sum_{i\in A} \frac{v_i\alpha_i}{v_i+\alpha_i} + \sum_{i\in B} \frac{v_i\alpha_i}{v_i+\alpha_i} + v_{k+1} \\
    \le& 2\cdot \left(\sum_{i\in A} \alpha_i\cdot\left(1-\frac{\alpha_i}{v_{k+1}+\alpha_i}\right) + \frac{1}{2}\cdot \sum_{i\in B} \min\set{v_i,\alpha_i}\right) + v_{k+1} \\
    \le& 2\cdot\alg + \alg = 3\cdot\alg.
  \end{aligned}
\end{equation*}
The first inequality comes from the same argument as in Case 1.
The second inequality holds because of \eqref{eq:approx-lw:case2:vi} and the fact that $\frac{v_i\alpha_i}{v_i+\alpha_i}\le \min\{v_i,\alpha_i\}$ for all $i\in [n]$.
The last inequality holds because of \eqref{eq:approx-lw:case2:alg} and \eqref{eq:approx-lw:case2:alg2}.

Combining the two cases, we have $\alg \ge 1/3\cdot \opt$.
\end{proof}

Combining \cref{sol:approx-lw:dummy,sol:approx-lw:approx,sol:approx-lw:monotonic,sol:approx-lw:budget}, along with Myerson's Lemma (\cref{lem:myerson}), directly proves \cref{thm:approx-lw}.

\subsection{Upper Bound on the Approximation Ratio}

We now show that there is an upper bound on the approximation ratio with respect to liquid welfare that any truthful mechanism can achieve.

\begin{lemma}[$1/2$-Approximation Upper Bound]\label{sol:approx-lw:upper-bound}
  It is impossible for any truthful mechanism to achieve an approximation ratio better than $1/2$ for liquid welfare.
\end{lemma}

\begin{proof}
  Consider any truthful mechanism that achieves an approximation ratio of $\rho$ for liquid welfare. Suppose there are two bidders with valuations $v_1$ and $v_2$, and budget impact factors $\alpha_1$ and $\alpha_2$, where $v_2 = \alpha_2 = 1$ and $\alpha_1> 1$. Let $x_1(v_1, v_2)$ and $x_2(v_1, v_2)$ denote the allocations for bidders 1 and 2, respectively.
  
  Now consider the scenario in which bidder 1's valuation is $v_1 = \alpha_1^2 > v_2$.
  From \cref{alg:allocation}, the optimal allocation assigns $x_1^* = \frac{1}{\alpha_1 + 1}<\frac{1}{2}$ to bidder 1 and $x_2^* = 1 - \frac{1}{\alpha_1 + 1}>x_1^*$ to bidder 2. 
  The optimal liquid welfare is 
  \begin{equation*}
    \opt = \min\set*{\alpha_1^2 x_1^*, \alpha_1 x_2^*} + \min\set*{x_2^*, x_1^*} = \frac{\alpha_1^2 + 1}{\alpha_1 + 1}. 
  \end{equation*} 
  Because the mechanism is $\rho$-approximation, we have 
  \begin{equation*}
    \min\set*{\alpha_1^2x_1(\alpha_1^2,1),\alpha_1x_2(\alpha_1^2,1)}+\min\set*{x_2(\alpha_1^2,1),x_1(\alpha_1^2,1)}\ge \rho\cdot\opt.
  \end{equation*}
  Additionally, 
  \begin{equation*}
    \min\set*{\alpha_1^2x_1(\alpha_1^2,1),\alpha_1x_2(\alpha_1^2,1)}+\min\set*{x_2(\alpha_1^2,1),x_1(\alpha_1^2,1)} \le (\alpha_1+1)\cdot x_2(\alpha_1^2,1).
  \end{equation*}
  Therefore, we have $x_2(\alpha_1^2,1) \ge \rho \cdot \frac{\alpha_1^2 + 1}{(\alpha_1+1)^2}$.
  As $x_1(\alpha_1^2,1)+x_2(\alpha_1^2,1)\le 1$, we have $x_1(\alpha_1^2,1) \le 1 - \rho \cdot \frac{\alpha_1^2 + 1}{(\alpha_1+1)^2}$.

  Consider another scenario in which bidder 1's valuation is $v_1 = \sqrt{\alpha_1} > v_2$.
  The optimal allocation assigns $x_1'^* = \frac{\sqrt{\alpha_1}}{\sqrt{\alpha_1} + 1}>\frac{1}{2}$ to bidder 1 and $x_2'^* = 1 - \frac{\sqrt{\alpha_1}}{\sqrt{\alpha_1} + 1}<x_1'^*$ to bidder 2.
  The optimal liquid welfare is 
 $
    \opt' = \min\set*{\sqrt{\alpha_1} x_1'^*, \alpha_1 x_2'^*} + \min\set*{x_2'^*, x_1'^*} = \frac{\alpha_1 + 1}{\sqrt{\alpha_1} + 1}. 
 $ 
  By the same argument, we have
  \begin{equation*}
    \begin{aligned}
      &(\sqrt{\alpha_1}+1)\cdot x_1(\sqrt{\alpha_1},1) \\
      \ge& \min\set*{\sqrt{\alpha_1}x_1(\sqrt{\alpha_1},1),\alpha_1x_2(\sqrt{\alpha_1},1)}+\min\set*{x_2(\sqrt{\alpha_1},1)+x_1(\sqrt{\alpha_1},1)}\\
      \ge& \rho\cdot\opt' = \rho\cdot \frac{\alpha_1 + 1}{\sqrt{\alpha_1} + 1}.
    \end{aligned}
  \end{equation*}
  Therefore, we have $x_1(\sqrt{\alpha_1},1) \ge \rho\cdot \frac{\alpha_1 + 1}{(\sqrt{\alpha_1} + 1)^2}$.
  Since the mechanism is truthful and $\alpha_1^2 > \sqrt{\alpha_1}$, it follows from Myerson's Lemma (\cref{lem:myerson}) that $x_1(\alpha_1^2,1) \ge x_1(\sqrt{\alpha_1},1)$.
  Therefore, we have
  \begin{equation*}
    1 - \rho \cdot \frac{\alpha_1^2 + 1}{(\alpha_1+1)^2}
    \ge
    \rho\cdot \frac{\alpha_1 + 1}{(\sqrt{\alpha_1} + 1)^2}.
  \end{equation*}
  By simple algebra, we have
  \begin{equation*}
    \rho \le 
    1 \Big/ \left(
      \frac{\alpha_1^2 + 1}{(\alpha_1+1)^2} + 
      \frac{\alpha_1 + 1}{(\sqrt{\alpha_1} + 1)^2}
    \right).
  \end{equation*}
  Let $\alpha_1 \to \infty$, we have $\rho \le 1/2$.
\end{proof}

\section{Conclusion}\label{sec:conclusion}

This paper studies a novel problem of the allocation-induced externality on budgets. 
This issue can be regarded as a practical extension of budget constraints in auctions and also represents a new category of allocation externality problems. 
For this problem, we first propose an allocation algorithm designed to maximize liquid welfare, a natural objective in budgeted settings, from the perspective of algorithm design.  
While this algorithm achieves optimal welfare, it lacks monotonicity, which is a challenge for us to obtain a truthful mechanism.
Therefore, we are motivated by the mechanism design for the uniform-price auction, and adaptively assign items to ensure the monotonicity of the allocation rule while accounting for budget adjustments. 
This approach enables us to design a truthful and individually rational auction mechanism that provides a $1/3$-approximation for the liquid welfare. 
Moreover, we derive an upper bound of the approximation ratio by proving that no truthful mechanism can achieve an approximation ratio better than $1/2$.
Closing this approximation gap is an interesting direction for future work. 
Additionally, there are many practically significant extensions of allocation externalities on budgets that merit further investigation, such as heterogeneous externality impacts and considering externalities on both valuations and budgets simultaneously. Addressing these problems requires new mechanism design and analytical techniques.

\newpage
\bibliographystyle{splncs04}
\bibliography{reference}

\newpage
\appendix

\section{Omitted Proofs in \cref{sec:optimal}}\label{sec:proof-optimal}

\subsection{Proof of \cref{thm:opt-allocation}}\label{sec:proof-optimal-thm}

We begin by constructing an optimal allocation that adheres to properties (P1)-(P4) as outlined in the proof sketch.

We first construct an optimal allocation that satisfies (P1).
Consider any optimal allocation $\vecx=(x_1,x_2,\cdots,x_n)$ that achieves the optimal liquid welfare.
If $\sum_{i\in[n]} x_i=1$, then $\vecx$ is the desired allocation.
Otherwise, let $\epsilon\deq 1-\sum_{i\in[n]} x_i$, and allocate $\epsilon$ to arbitrary bidder $i$.
Then bidder $i$'s value $v_ix_i$ will increase $v_i\epsilon$, while $i$'s budget $\alpha_i(1-(x_i+\epsilon))=\alpha_i(1-\epsilon-x_i)$ remains unchanged.
In the meanwhile, the value of other bidder $j$ will not be changed, and their budget will increase $\alpha_i\epsilon$.
Therefore, the liquid welfare will not decrease.
Thus, the new allocation is optimal and allocates all items.

Next, we construct an optimal allocation that satisfies (P1) and (P2).
Denote $\vecx=(x_1,x_2,\cdots,x_n)$ as the optimal allocation that satisfies (P1), i.e., $\sum_{i\in[n]} x_i=1$.
Thus, the budget of bidder $i$ is $B_i(\vecx_{-i})=\alpha_i\cdot\sum_{j\in[n]-i} x_j=\alpha_i\cdot (1-x_i)$.
If $\vecx$ satisfies (P2), then $\vecx$ is the desired allocation.
Otherwise, let $i\neq \l$ be the first bidder that violates (P2), i.e., $x_i > \frac{\alpha_i}{v_i+\alpha_i}$.
Let $\epsilon\deq x_i - \frac{\alpha_i}{v_i+\alpha_i}$, and construct a new allocation $\vecx'=(x_1',x_2',\cdots,x_n')$ by reallocating $\epsilon$ from bidder $i$ to bidder $\l$, i.e., $x_i'=x_i-\epsilon$, $x_\l'=x_\l+\epsilon$ and $x_k'=x_k$ for all $k\neq i,\l$.
Then the change of liquid welfare is
\begin{equation*}
  \begin{aligned}
    \lw(\vecx')-\lw(\vecx) =& \min\set{v_ix_i',\alpha_i(1-x_i')} - \min\set{v_ix_i,\alpha_i(1-x_i)} \\
    &\quad + \min\set{v_\l x_\l',\alpha_\l(1-x_\l')} - \min\set{v_\l x_\l,\alpha_\l(1-x_\l)} \\
    \ge& \alpha_i\epsilon - \alpha_\l\epsilon \ge 0.
  \end{aligned}
\end{equation*}
The last inequality holds because $\alpha_i\ge \alpha_\l$.
After the above item reallocation, we have $x_i'=\frac{\alpha_i}{v_i+\alpha_i}$.
Thus, each execution of this process reduces at least one violating bidder of (P2), and the liquid welfare does not decrease.
Therefore, by repeating the above process, we can ultimately obtain an optimal allocation that satisfies (P1) and (P2).

Then we construct an optimal allocation that satisfies (P1)-(P3).
Let $\vecx=(x_1,x_2,\cdots,x_n)$ be the optimal allocation that satisfies (P1) and (P2).
If $\vecx$ satisfies (P3), then $\vecx$ is the desired allocation.
Otherwise, let $i$ and $j$ be the pair of bidders that violates (P3), where $i$ is the smallest index and $j$ is the largest index. 
Specifically, $i \deq \min\set*{i \mid x_i < \frac{\alpha_i}{v_i + \alpha_i}}$ and $j \deq \max\set*{j \mid x_j > 0}$.
Let $\epsilon\deq \min\set*{x_j,\frac{\alpha_i}{v_i+\alpha_i}-x_i}$, and construct a new allocation $\vecx'=(x_1',x_2',\cdots,x_n')$ by reallocating $\epsilon$ from bidder $j$ to bidder $i$, i.e., $x_i'=x_i+\epsilon$, $x_j'=x_j-\epsilon$ and $x_t'=x_t$ for all $t\neq i,j$.
Then the change of liquid welfare is
\begin{equation*}
  \begin{aligned}
    \lw(\vecx')-\lw(\vecx) =& \min\set{v_ix_i',\alpha_i(1-x_i')} - \min\set{v_ix_i,\alpha_i(1-x_i)} \\
    &\quad + \min\set{v_jx_j',\alpha_j(1-x_j')} - \min\set{v_jx_j,\alpha_j(1-x_j)} \\
    \ge& v_i\epsilon - v_j\epsilon \ge 0.
  \end{aligned}
\end{equation*}
The last inequality holds because $v_i \ge v_j$. 
After the above item reallocation, either $x_i' = \frac{\alpha_i}{v_i + \alpha_i}$ and $0\le x_j'<x_j$ or $x_j' = 0$ and $x_i' \le \frac{\alpha_i}{v_i + \alpha_i}$.
Thus, each execution of this process reduces at least one violating bidder pair of (P3), does not cause any bidder to violate (P2), and does not decrease the liquid welfare.
Therefore, by repeating the above process, we can ultimately obtain an optimal allocation that satisfies (P1)-(P3).

Finally, we construct an optimal allocation that satisfies (P1)-(P4).
Let $\vecx=(x_1,x_2,\cdots,x_n)$ be the optimal allocation that satisfies (P1)-(P3).
If $\vecx$ satisfies (P4), then $\vecx$ is the desired allocation.
Otherwise, let $i$ be the first bidder that violates (P4).
Specifically, $i\deq\min\set*{i\mid x_i < \frac{\alpha_i}{v_i+\alpha_i}}$ and $x_\l > \frac{\alpha_\l}{v_\l+\alpha_\l}$.
Since $\vecx$ satisfies (P2) and (P3), $x_\l >0$ implies that $x_j = \frac{\alpha_j}{v_j + \alpha_j}$ for all $j < \l$. 
Furthermore, since $i$ is the smallest index such taht $x_i<\frac{\alpha_i}{v_i+\alpha_i}$, it follows that $i > \l$.
Additionally, for all $\l < j < i$, we have $x_j = \frac{\alpha_j}{v_j + \alpha_j}$, and for all $j > i$, we have $x_j = 0$.
Let $\epsilon \deq \min\set*{x_\l-\frac{\alpha_\l}{v_\l+\alpha_\l}, \frac{\alpha_i}{v_i+\alpha_i}-x_i}$, and construct a new allocation $\vecx'=(x_1',x_2',\cdots,x_n')$ by reallocating $\epsilon$ from bidder $\l$ to bidder $i$, i.e., $x_i'=x_i+\epsilon$, $x_\l'=x_\l-\epsilon$ and $x_t'=x_t$ for all $t\neq i,\l$.
Then the change of liquid welfare is
\begin{equation*}
  \begin{aligned}
    \lw(\vecx')-\lw(\vecx) =& \min\set{v_ix_i',\alpha_i(1-x_i')} - \min\set{v_ix_i,\alpha_i(1-x_i)} \\
    &\quad + \min\set{v_\l x_\l',\alpha_\l(1-x_\l')} - \min\set{v_\l x_\l,\alpha_\l(1-x_\l)} \\
    =& v_i\epsilon + \alpha_\l\epsilon \ge 0.
  \end{aligned}
\end{equation*}
After the above item reallocation, either $x_i' = \frac{\alpha_i}{v_i + \alpha_i}$ and $\frac{\alpha_\l}{v_\l + \alpha_\l}\le x_\l'<x_\l$ or $x_\l' = \frac{\alpha_\l}{v_\l + \alpha_\l}$ and $x_i' \le \frac{\alpha_i}{v_i + \alpha_i}$.
Thus, each execution of this process reduces at least one violating bidder of (P4) or directly makes the new allocation satisfy (P4). 
Moreover, the process does not cause any bidder to violate (P2) or (P3) and does not decrease the liquid welfare.
Therefore, by repeating the above process, we can ultimately obtain an optimal allocation that satisfies (P1)-(P4).

Next, we prove that the allocation satisfying properties (P1)-(P4) is unique.

We proceed by contradiction. Assume that there exist two allocations $\vecx = (x_1, x_2, \cdots, x_n)$ and $\vecx' = (x_1', x_2', \cdots, x_n')$ that satisfy (P1)-(P4). Let $i$ be the smallest index such that $x_i \neq x_i'$. Without loss of generality, assume $0 \le x_i < x_i'$.
Since $\vecx$ and $\vecx'$ satisfy (P1), there must exist a bidder $j > i$ such that $x_j > x_j' \ge 0$. Due to (P3), $x_i' > 0$ implies that $x_t' \ge \frac{\alpha_t}{v_t + \alpha_t}$ for all $t < i$. Similarly, due to (P3), $x_j > 0$ implies that $x_t \ge \frac{\alpha_t}{v_t + \alpha_t}$ for all $t < j$. Thus, $x_i' > x_i \ge \frac{\alpha_i}{v_i + \alpha_i}$.
Due to (P2), we have $i = \ell$. It follows that $x_\ell' > \frac{\alpha_i}{v_i + \alpha_i}$. Together with (P2) and (P4), we have $x_t' = \frac{\alpha_t}{v_t + \alpha_t}$ for all $t \neq \ell$. Therefore, $x_j > x_j' = \frac{\alpha_j}{v_j + \alpha_j}$. Since $j \neq \ell$, this contradicts (P2).
Therefore, the allocation satisfying properties (P1)-(P4) is unique.

Clearly, the allocation provided by \cref{alg:allocation} satisfies (P1)-(P4).
Since the allocation satisfies (P1)-(P4) is unique, and there is at least one optimal allocation that satisfies (P1)-(P4), the allocation provided by \cref{alg:allocation} is optimal.

\section{Omitted Proofs in \cref{sec:mechanism}}\label{sec:proof-mechanism}

\subsection{Proof of \cref{sol:approx-lw:dummy}}\label{proof:approx-lw:dummy}

We consider three cases.

\textit{Case 1:} $\sum_{i=1}^n \min\set*{\frac{\alpha_i}{v_n+\alpha_i},\frac{1}{2}} > 1$.
Thus, by the definition of $k$, we have $k < n$. Adding a dummy bidder with $v_{n+1} = 0$ and $\alpha_{n+1} > 0$ will not affect the allocation.

\textit{Case 2:} $n=2$ and $\sum_{i=1}^n \min\set*{\frac{\alpha_i}{v_n+\alpha_i},\frac{1}{2}} \le 1$.
After adding the dummy bidder $n+1=3$, 
we have $\sum_{i=1}^3 \min\set*{\frac{\alpha_i}{v_3+\alpha_i},\frac{1}{2}} = \sum_{i=1}^3 \min\set*{\frac{\alpha_i}{\alpha_i},\frac{1}{2}} = \frac{3}{2} > 1$,
therefore the division point is 2.
Since $\sum_{i=1}^2 \min\set*{\frac{\alpha_i}{0+\alpha_i},\frac{1}{2}}=1$, we have $q=0=v_3$, thus $x_3=1-\sum_{i=1}^2 \min\set*{\frac{\alpha_i}{\alpha_i},\frac{1}{2}} = 1-1=0$.

\textit{Case 3:} $n>2$ and $\sum_{i=1}^n \min\set*{\frac{\alpha_i}{v_n+\alpha_i},\frac{1}{2}} \le 1$.
After adding the dummy bidder $n+1$, we have $\sum_{i=1}^{n+1} \min\set*{\frac{\alpha_i}{v_{n+1}+\alpha_i},\frac{1}{2}} = \sum_{i=1}^{n+1} \min\set*{\frac{\alpha_i}{\alpha_i},\frac{1}{2}} > \frac{3}{2} > 1$,
therefore the division point is $n$.
Since $\sum_{i=1}^n \min\set*{\frac{\alpha_i}{0+\alpha_i},\frac{1}{2}} \ge \frac{3}{2} > 1$,
we have $q>0=v_{n+1}$.
Thus $x_{n+1}=0$.

\subsection{Proof of \cref{sol:approx-lw:monotonic}}\label{proof:approx-lw:monotonic}

Let $k$ and $q$ denote the division point and the uniform price computed in \cref{mech:approx-lw}, respectively.
For bidders $j=1,2,\cdots,k$, $x_j$ remains unchanged if $v_j$ increases.
For bidders $j=k+2,\cdots,n$, $x_j$ remains at 0 if $v_j$ increases but still satisfies $v_j \le v_{k+1}$.
Thus, we only need to consider bidder $k+1$.

If $v_{k+1}<q$, we have $x_{k+1}=0$.
Let $q'$ be the smallest root of the equation $\sum_{i=1}^{k+1} \min\set*{\frac{\alpha_i}{q'+\alpha_i},\frac{1}{2}}=1$.
Clearly, $q' > q$ since $\alpha_{k+1}>0$, and $v_k \ge q$ by the definition of $k$ and $q$.
As $v_{k+1}$ increases from $q$ to $\min\set*{v_k,q'}$, $x_{k+1}=1-\sum_{i=1}^k\min\set*{\frac{\alpha_i}{v_{k+1}+\alpha_i},\frac{1}{2}}$, which is non-decreasing in $v_{k+1}$.
When $v_{k+1} = \min\set{v_k, q'}$, two cases need to be considered.

\textit{Case 1:} $q'\le v_k$.
In this case $v_{k+1}$ first increases to $v_{k+1}=q'$.
Therefore, we have
$\sum_{i=1}^{k+1}\min\set*{\frac{\alpha_i}{v_{k+1}+\alpha_i},\frac{1}{2}}=\sum_{i=1}^{k+1}\min\set*{\frac{\alpha_i}{q'+\alpha_i},\frac{1}{2}}=1$.
Additionally,
$\sum_{i=1}^{k+2}\min\set*{\frac{\alpha_i}{v_{k+2}+\alpha_i},\frac{1}{2}} > 1$ since $v_{k+2}\le v_{k+1}\le q'$ and $\alpha_{k+2}>0$.
Therefore, the division point will become $k+1$.
Thus, $x_{k+1}=\min\set*{\frac{\alpha_{k+1}}{q'+\alpha_{k+1}},\frac{1}{2}}=1-\sum_{i=1}^k\min\set*{\frac{\alpha_i}{q'+\alpha_i},\frac{1}{2}}$,
Therefore, $x_{k+1}$ changes continuously at $v_{k+1} = q'$.

\textit{Case 2:} $q' > v_k$.
In this case $v_{k+1}$ first increases to $v_{k+1}=v_k$. 
For simplicity, assume that ties in valuation are broken in favor of bidder $k+1$, i.e., bidder $k+1$ is ordered before bidder $k$ when $v_{k+1}=v_k$.
Thus we have
$\sum_{i=1}^{k+1}\min\set*{\frac{\alpha_i}{v_{k+1}+\alpha_i},\frac{1}{2}}=\sum_{i=1}^{k+1}\min\set*{\frac{\alpha_i}{v_k+\alpha_i},\frac{1}{2}}> \sum_{i=1}^{k+1}\min\set*{\frac{\alpha_i}{q'+\alpha_i},\frac{1}{2}}=1$.
The inequality holds because $q' > v_k$ and $q'$ is the smallest root of the final equation.
Therefore, the division point will be less than or equal to $k$. 
Due to the definition of $k$, the division point should be at least $k-1$. 
Thus, the division point could be either $k-1$ or $k$. We need to discuss the following two subcases.

\textit{Case 2-1:} $\sum_{i=1}^{k-1}\min\set*{\frac{\alpha_i}{v_k+\alpha_i},\frac{1}{2}}+\min\set*{\frac{\alpha_{k+1}}{v_k+\alpha_{k+1}},\frac{1}{2}}\le 1$, the division point is $k$.
Let $q''$ be the smallest root of $\sum_{i=1}^{k-1}\min\set*{\frac{\alpha_i}{q''+\alpha_i},\frac{1}{2}}+\min\set*{\frac{\alpha_{k+1}}{q''+\alpha_{k+1}},\frac{1}{2}} = 1$.
Therefore $q''\le v_k$.
Thus, $x_{k+1}=\min\set*{\frac{\alpha_{k+1}}{v_k+\alpha_{k+1}},\frac{1}{2}} > 1-\sum_{i=1}^k\min\set*{\frac{\alpha_i}{v_k+\alpha_i},\frac{1}{2}}$, 
The inequality holds because $q' > v_k$ in Case 2.
Therefore, $x_{k+1}$ is strictly increasing at $v_{k+1} = v_k$.

\textit{Case 2-2:} $\sum_{i=1}^{k-1}\min\set*{\frac{\alpha_i}{v_k+\alpha_i},\frac{1}{2}}+\min\set*{\frac{\alpha_{k+1}}{v_k+\alpha_{k+1}},\frac{1}{2}}> 1$, the division point is $k-1$.
Let $q'''$ be the smallest root of $\sum_{i=1}^{k-1}\min\set*{\frac{\alpha_i}{q'''+\alpha_i},\frac{1}{2}} = 1$.
Since $\sum_{i=1}^{k-1}\min\set*{\frac{\alpha_i}{v_k+\alpha_i},\frac{1}{2}} \le \sum_{i=1}^k\min\set*{\frac{\alpha_i}{v_k+\alpha_i},\frac{1}{2}} \le 1$ due to the definition of $k$, we have $q'''\le v_k$.
Thus $x_{k+1}=1-\sum_{i=1}^{k-1}\min\set*{\frac{\alpha_i}{v_k+\alpha_i},\frac{1}{2}}\ge 1-\sum_{i=1}^k\min\set*{\frac{\alpha_i}{v_k+\alpha_i},\frac{1}{2}}$, 
$x_{k+1}$ is non-decreasing at $v_{k+1} = v_k$.

\subsection{Proof of \cref{sol:approx-lw:budget}}\label{proof:approx-lw:budget}

\noindent{\bf Case 1:} $q > v_{k+1}$. In this case, only bidder $i \in[k]$ receives the item. 
Thus, it is enough for us to check the payment for $i \in [k]$. 
For bidders $i \in [k]$, $x_i = \min\set*{\frac{\alpha_i}{q + \alpha_i}, \frac{1}{2}} \le \frac{\alpha_i}{q + \alpha_i}$, therefore $q x_i \le \alpha_i (1 - x_i)$.
Since the allocation for bidder $i$ does not change with $v_i$ as long as $v_i > q$, and by \cref{lem:myerson}, we have $p_i \le q x_i \le \alpha_i (1 - x_i) = B_i$.

\noindent{\bf Case 2:} $q \le v_{k+1}$.
Only bidders $i = 1, 2, \cdots, k+1$ can receive the item. 
For bidders $i=1,2,\cdots, k$, using the same argument as in Case 1, we have $p_i\le v_{k+1} x_i \le \alpha_i(1-x_i)=B_i$.
Now we consider $x_{k+1}$. 
According to \eqref{eq:approx-lw:xkplus1}, we have $x_{k+1} \le \frac{\alpha_{k+1}}{v_{k+1} + \alpha_{k+1}}$, implying $v_{k+1} x_{k+1} \le \alpha_{k+1} (1 - x_{k+1})$. 
Thus, along with \cref{lem:myerson}, $p_{k+1} \le v_{k+1} x_{k+1} \le \alpha_{k+1} (1 - x_{k+1}) = B_{k+1}$.

\section{Myerson's Lemma}\label{sec:myerson}

In this section, we present Myerson's lemma, which characterizes the allocation and payment of a truthful and individually rational auction mechanism.

\begin{lemma}[Myerson's Lemma~\cite{myerson1981optimal}]\label{lem:myerson}
An auction mechanism is truthful and individually rational if and only if the allocation and payment functions of the mechanism satisfy the following conditions:

\begin{enumerate}[left=1em]
  \item[(1)] The allocation function $x_i(v_i,\vecv_{-i})$ is non-decreasing in $v_i$ for all $i \in [n]$ and any $\vecv_{-i}$.
  \vspace{0.5em}
  \item[(2)] The payment function is given by $p_i(v_i,\vecv_{-i}) = v_i x_i(v_i,\vecv_{-i}) - \int_{0}^{v_i} x_i(z,\vecv_{-i}) \dif z$ for all $i \in [n]$.
\end{enumerate}
\end{lemma}

\end{document}